\documentclass[12pt,a4paper]{article}

\usepackage{empheq}
\usepackage[us,nodayofweek]{datetime}
\usepackage[utf8]{inputenc}
\usepackage[T1]{fontenc}
\usepackage{amsmath,amssymb,amsfonts,amsthm}
\usepackage{mathtools}
\usepackage{bm}
\usepackage[margin=1in]{geometry}
\usepackage{graphicx}
\usepackage{multirow}
\usepackage{booktabs}
\usepackage{tabularx}
\usepackage{natbib}
\usepackage{hyperref}
\usepackage{setspace}
\usepackage{float}
\usepackage{caption}
\usepackage{subcaption}
\usepackage{xcolor}
\usepackage{enumitem}
\usepackage{tikz}
\usetikzlibrary{patterns}
\usepackage{blkarray} 
\newtheorem{theorem}{Theorem}
\newtheorem{lemma}{Lemma}

\title{\Large\textbf{Simultaneous Clustered Orthogonalization}}

\author{
Bastien Buchwalter\thanks{SKEMA Business School \& Universit\'{e} C\^{o}te d'Azur. Email: \texttt{bastien.buchwalter@skema.edu}}
\and
Francis X.\ Diebold\thanks{University of Pennsylvania \& NBER. Email: \texttt{fdiebold@sas.upenn.edu}}
\and
Kamil Y{\i}lmaz\thanks{Ko\c{c} University. Email: \texttt{kyilmaz@ku.edu.tr}}
}

\date{\normalsize \vspace{4mm} First Draft: August 2026\\ This Draft: \today\ }

\begin{document}

\maketitle

\begin{abstract}
Identification in vector autoregressions involves two distinct choices: how extensively to orthogonalize shocks and whether orthogonality is imposed sequentially or simultaneously. Generalized identification imposes no orthogonality; \citet{Sims1980} imposes full orthogonality sequentially; \citet{francis2026principled} impose full orthogonality simultaneously; and \citet{BDY1} provide clustered partial orthogonalization sequentially. We fill the remaining case by developing \emph{simultaneous clustered orthogonalization} (SCO). SCO preserves contemporaneous dependence within economically meaningful clusters while imposing orthogonality across clusters jointly, thereby eliminating dependence on cluster ordering. We formulate the associated correlation-maximizing identification problem and show that, in correlation space, it reduces to the quadratic matrix equation $\boldsymbol Z \boldsymbol R_1 \boldsymbol Z = \boldsymbol R_C$. This yields a closed-form solution, and we prove that the associated identification matrix is the unique global maximizer. SCO is order- and scale-invariant and nests generalized identification and full simultaneous orthogonalization as special cases, yielding a flexible family of structural decompositions indexed by the number and composition of clusters.
\end{abstract}

\bigskip
\noindent\textbf{Acknowledgments:} Kamil Y{\i}lmaz gratefully acknowledges the support of the Scientific and Technological Research Council of Turkey (T\"{U}B\.{I}TAK) under grant number 121C271.



\thispagestyle{plain}
\thispagestyle{empty}

\newpage

\setcounter{page}{1}
\thispagestyle{empty}

\section{Introduction}

Identification is fundamental to structural analysis of vector autoregressions and, in particular, to connectedness measurement based on impulse responses and forecast error variance decompositions. When reduced-form innovations are contemporaneously correlated, their effects cannot in general be uniquely attributed to distinct shocks. Identification therefore requires restrictions on contemporaneous dependence. In connectedness analysis, those restrictions are central to distinguishing shock transmission from contemporaneous co-movement under the maintained identification scheme.

The relevant identification schemes can be organized along two dimensions. The first is the \emph{extent} of orthogonalization. At one extreme, generalized identification preserves the full reduced-form covariance structure and imposes no orthogonality. At the other, full orthogonalization makes all structural shocks mutually orthogonal. Between those extremes, \citet{BDY1} introduce clustered identification: shocks may remain correlated within economically meaningful clusters but are orthogonal across clusters. This provides a natural middle ground and, under sequential identification, sharply reduces the ordering problem. With $N$ variables, full sequential orthogonalization admits $N!$ orderings, whereas clustered sequential orthogonalization requires consideration of only $C!$ cluster orderings.

The second dimension is \emph{how} orthogonality is imposed. The clustered framework of \citet{BDY1} is sequential: clusters are orthogonalized one after another. Clustering therefore reduces but does not eliminate ordering dependence. Different orderings of the same clusters can produce different structural shocks and hence different connectedness measures, because clusters appearing earlier may absorb contemporaneous variation shared with clusters appearing later. \citet{BDY2} document this sensitivity empirically by examining the distribution of cluster-level connectedness across admissible orderings.

\citet{francis2026principled} take a different approach to the ordering dimension. Rather than identifying shocks sequentially, they formulate identification as a simultaneous optimization problem. Under full orthogonality, their procedure selects, among all admissible structural decompositions, the one that maximizes the correspondence between reduced-form and structural shocks. The resulting identification is order- and scale-invariant. It therefore removes the ordering problem, but only by imposing the strongest covariance restriction: mutual orthogonality of all structural shocks.

Viewed in this two-dimensional space, one important case is missing: \emph{partial} orthogonalization implemented \emph{simultaneously}. This paper fills that gap. We retain the economically meaningful clustered structure of \citet{BDY1}, but replace sequential cross-cluster orthogonalization with simultaneous optimization. We call the resulting procedure \emph{simultaneous clustered orthogonalization} (SCO). SCO permits contemporaneous dependence within clusters while imposing orthogonality across clusters jointly, thereby preserving the economic role of clustering and eliminating dependence on cluster ordering.

The resulting identification problem has a particularly simple representation in correlation space. We show that it reduces to the quadratic matrix equation $\boldsymbol Z \boldsymbol R_1 \boldsymbol Z = \boldsymbol R_C$, where $\boldsymbol R_1$ is the reduced-form correlation matrix and $\boldsymbol R_C$ is its clustered target counterpart. For symmetric positive definite $\boldsymbol R_1$ and $\boldsymbol R_C$, the equation has a closed-form symmetric positive definite solution, and the associated identification matrix is the unique global maximizer of the underlying correlation criterion. The framework nests generalized identification when all variables form a single cluster and the full simultaneous orthogonalization of \citet{francis2026principled} when every variable forms its own cluster. Intermediate partitions generate a family of order-invariant partial orthogonalizations indexed by the number and composition of clusters.

The remainder of the paper develops the framework. Section 2 introduces the VAR, impulse-response, and variance-decomposition setup. Section 3 organizes existing identification schemes along the two dimensions of extent and implementation of orthogonality. Section 4 formulates the SCO optimization problem and derives its first-order conditions. Section 5 solves the problem in correlation space and establishes global optimality. Section 6 derives the generalized and fully orthogonalized schemes as corner solutions. Section 7 concludes.

\section{Framework}

We work in a standard vector autoregression (VAR) framework and measure connectedness using impulse response functions (IRFs) and forecast error variance decompositions (VDs). Because reduced-form residuals may be contemporaneously correlated, additional identifying structure is needed to attribute innovations to distinct shocks and trace their transmission through the system.

\subsection{Vector Autoregression}

Consider an $N$-variable, covariance-stationary VAR($p$) of the form $\mathbf{x}_t=\sum_{i=1}^{p}\boldsymbol{\Phi}_i \mathbf{x}_{t-i}+\mathbf{u}_t,$
where $\mathbb{E}[\mathbf{u}_t]=0$ and $\mathbb{V}[\mathbf{u}_t]=\boldsymbol{\Sigma}$ for all $t$. Under covariance stationarity, the system admits a moving-average representation,
$\mathbf{x}_t=\sum_{i=0}^{\infty}\boldsymbol{A}_i\,\mathbf{u}_{t-i}$,
where $\boldsymbol{A}_0 = \boldsymbol I_N$ and, for $i \ge 1$, the coefficient matrices satisfy the recursion $\boldsymbol{A}_i=\boldsymbol{\Phi}_1\boldsymbol{A}_{i-1}+\cdots+\boldsymbol{\Phi}_p\boldsymbol{A}_{i-p}$,
with $\boldsymbol{A}_{i-p}=\mathbf{0}$ for $i-p<0$.

Because the reduced-form residuals $\mathbf{u}_t$ are typically contemporaneously correlated, their effects cannot be uniquely attributed to distinct shocks without further restrictions. We therefore introduce structural shocks through a nonsingular transformation. Let $\boldsymbol{Q}_C$ be any nonsingular matrix, and define $\boldsymbol{\epsilon}_t = \boldsymbol{Q}_C^{-1}\mathbf{u}_t$, where $C$ denotes the number of clusters. It follows that $\mathbb{E}[\boldsymbol{\epsilon}_t]=0$ and $\mathbb{V}[\boldsymbol{\epsilon}_t]=\boldsymbol{\Omega}_C$, where $\boldsymbol{\Omega}_C$ denotes a pre-specified block-diagonal target covariance matrix encoding the desired clustering structure. The moving-average representation can be rewritten as

\begin{align*}
 \mathbf{x}_t=\sum_{i=0}^{\infty}(\boldsymbol{A}_i \boldsymbol{Q}_C)\,\boldsymbol{\epsilon}_{t-i}.
\end{align*}
The matrix $\boldsymbol{Q}_C$ determines how the system's dynamics are decomposed into within-cluster co-movement and cross-cluster shock transmission. In the sequential framework of \citet{BDY1}, $\boldsymbol{Q}_C$ is block lower-triangular by construction, reflecting the order in which clusters are orthogonalized. Here, by contrast, we leave $\boldsymbol{Q}_C$ unrestricted and determine it by optimization. As shown in Section~\ref{sec:corrspace}, the resulting matrix is symmetric up to diagonal scalings rather than triangular, and the identified shocks are invariant to the ordering of variables or clusters. Thus $\boldsymbol{Q}_C$ separates within-cluster co-movement from cross-cluster shock transmission, allowing the latter to be studied without contemporaneous correlation across clusters.

This identification matters because orthogonality allows variance decompositions to attribute forecast-error variation to shocks associated with distinct variables or clusters. Without such restrictions, contemporaneously correlated reduced-form innovations cannot generally be assigned uniquely to individual shocks. This is the identification issue underlying the distinction between generalized and orthogonalized impulse responses and variance decompositions \citep{koop1996impulse,pesaran1998generalized}. In connectedness analysis, the same distinction underlies the separation of co-movement from contagion \citep{forbes2002no}.

\subsection{IRFs and VDs}
IRFs trace the system's dynamic response to a given structural shock. Following \citet{BDY1}, we write the scaled impulse response to a one-standard-deviation shock in $\epsilon_{jt}$ as
\begin{equation}
\boldsymbol{\psi}_j^C(h) = \frac{(\boldsymbol{A}_h\boldsymbol{Q}_C)\boldsymbol{\Omega}_C\boldsymbol{e}_j}{\sqrt{\omega_{C,jj}}},\nonumber
\label{scaled2}
\end{equation}
We denote the $H$-step-ahead VD by $\widetilde{\theta}_{ij}^C(H)$:
\begin{align}
\widetilde{\theta}_{ij}^C(H) = \frac{\displaystyle\sum_{h=0}^{H-1}\left(\boldsymbol{e}_i'\boldsymbol{\psi}_j^C(h)\right)^2}{\displaystyle\sum_{h=0}^{H-1}\boldsymbol{e}_i'\boldsymbol{A}_h\boldsymbol{\Sigma}\boldsymbol{A}_h'\boldsymbol{e}_i} = \frac{\displaystyle\omega_{C,jj}^{-1}\sum_{h=0}^{H-1}\left(\boldsymbol{e}_i'\boldsymbol{A}_h\boldsymbol{Q}_C\boldsymbol{\Omega}_C\boldsymbol{e}_j\right)^2}{\displaystyle\sum_{h=0}^{H-1}\boldsymbol{e}_i'\boldsymbol{A}_h\boldsymbol{\Sigma}\boldsymbol{A}_h'\boldsymbol{e}_i}, \nonumber
\label{thetasec}
\end{align}
where $\widetilde{\theta}_{ij}^C(H)$ is the share of the $H$-step-ahead forecast error variance of asset $i$ attributable to shocks originating from asset $j$. Because structural shocks may remain correlated under clustered and generalized identification, these raw shares need not sum to one: $\sum_{j=1}^N\widetilde{\theta}_{ij}^C(H)\neq 1$. Following \citet{diebold2012better}, we therefore normalize the variance decompositions as
\begin{equation}
\theta_{ij}^C(H) = \frac{\widetilde{\theta}_{ij}^C(H)}{\displaystyle\sum_{j=1}^N\widetilde{\theta}_{ij}^C(H)}.
\label{normalizedvd}
\end{equation}

For connectedness analysis, VDs provide a convenient summary of the dynamic responses by quantifying the contribution of individual shocks to forecast uncertainty. In particular, $\theta_{ij}^C(H)$ measures the share of the $H$-step-ahead forecast error variance of asset $i$ attributable to shocks originating from asset $j$. The resulting variance-decomposition matrix provides a natural basis for aggregating connectedness across assets and horizons and can also be interpreted as the weighted adjacency matrix of a directed network \citep{diebold2014network}.

\section{Orthogonalization of Residuals}

We now organize the identification schemes along the two dimensions introduced above: the extent of orthogonalization and whether orthogonality is imposed sequentially or simultaneously.

\subsection{Extent of Orthogonalization}

The first dimension concerns the extent of orthogonalization. At one extreme, with a single cluster, no orthogonalization is performed, corresponding to the generalized approach of \citet{koop1996impulse,pesaran1998generalized}. At the other extreme, with one variable per cluster, all structural shocks are mutually orthogonal, corresponding to the conventional Cholesky identification of \citet{Sims1980}. Between these extremes, \citet{BDY1} group variables into economically meaningful clusters, allowing shocks to remain correlated within clusters while imposing orthogonality across clusters.

Although clustering reduces the number of admissible orderings from $N!$ to $C!$, it does not eliminate ordering dependence: different cluster orderings can still produce different identified shocks and connectedness measures.

\subsection{Order of Orthogonalization}

Under sequential clustered identification, the clusters must still be ordered. As \citet{BDY2} show, the estimated role of a cluster can vary with its position in that ordering: a cluster appearing earlier may absorb a larger share of contemporaneous variation common across clusters. Ordering dependence is therefore economically relevant, not merely computational.

Simultaneous identification instead treats all clusters symmetrically and determines the structural shocks jointly. \citet{francis2026principled} provide such a scheme under full orthogonality. SCO combines their simultaneous approach with the clustered structure of \citet{BDY1}.

\subsection{Overview}

Table 1 summarizes the resulting framework. The rows distinguish the extent of orthogonalization, while the columns distinguish sequential from simultaneous identification, making the position of the present paper explicit.
   
\begin{table}[H]
\caption{Identification schemes along the two dimensions of orthogonalization}
\begin{center}
\renewcommand{\arraystretch}{1.5}
\setlength{\tabcolsep}{14pt}

\begin{tabular}{lcc}
\toprule
\textbf{Orthogonalization} & Sequential & Simultaneous \\ \midrule
None &
\multicolumn{2}{c}{\citet{koop1996impulse}; \citet{pesaran1998generalized}} \\
Partial & \citet{BDY1} & \textit{Present paper} \\
Full & \citet{Sims1980} & \citet{francis2026principled} \\
\bottomrule
\end{tabular}
\end{center}

\noindent \footnotesize
\textit{Notes:} ``None'' denotes generalized identification, under which the reduced-form covariance structure is preserved; in this case the distinction between sequential and simultaneous identification is not applicable. ``Partial'' denotes clustered identification, where structural shocks are orthogonal across clusters but may remain correlated within clusters. ``Full'' denotes mutual orthogonality of all structural shocks. Sequential identification imposes orthogonality according to an ordering, whereas simultaneous identification determines the structural shocks jointly and is therefore order-invariant.
\end{table}

\section{Simultaneous Clustered Orthogonalization}

We now develop SCO formally. We first formulate the optimization problem and then characterize the structural shocks under block-diagonal target covariance structures.

\subsection{Optimization Problem}

\citet{francis2026principled} cast structural identification as an optimization problem. They seek an identification matrix $\boldsymbol{Q}_N^{-1}$ such that the structural shocks are mutually orthogonal,
\begin{equation}
\mathbb{V}\!\left[\boldsymbol{Q}_N^{-1}\mathbf{u}_t\right]=\boldsymbol{I}_N,
\end{equation}
and, among all matrices satisfying this restriction, select the one that maximizes the correspondence between reduced-form and structural shocks,
\begin{align}
\max_{\boldsymbol{Q}_N^{-1}} \sum_{i=1}^{N}\operatorname{corr}(u_i,\epsilon_i), \qquad \text{s.t.} \qquad \mathbb{V}\!\left[\boldsymbol{Q}_N^{-1}\mathbf{u}_t\right]=\boldsymbol{I}_N.
\end{align}

SCO retains this identification criterion while generalizing the admissible covariance structure of the structural shocks. In particular, shocks may remain correlated within economically meaningful clusters while being orthogonal across clusters. Let $\boldsymbol{\Omega}_C$ denote the corresponding block-diagonal target covariance matrix. The simultaneous clustered orthogonalization problem is then
\begin{align*}
\max_{\boldsymbol{Q}_C^{-1}} \sum_{i=1}^{N}\operatorname{corr}(u_i,\epsilon_i), \qquad \text{s.t.} \qquad \mathbb{V}\!\left[\boldsymbol{Q}_C^{-1}\mathbf{u}_t\right]=\boldsymbol{\Omega}_C.
\end{align*}

For notational convenience, let $\boldsymbol{X}_C=\boldsymbol{Q}_C^{-1}$. The optimization problem can then be written as
\begin{align*}
\max_{\boldsymbol{X}_C} \sum_{i=1}^{N}\operatorname{corr}(u_i,\epsilon_i), \qquad \text{s.t.} \qquad \boldsymbol{X}_C\boldsymbol{\Sigma}\boldsymbol{X}_C^\top=\boldsymbol{\Omega}_C.
\end{align*}

Thus SCO preserves the correlation-based criterion of \citet{francis2026principled} while admitting intermediate covariance structures between full orthogonality and the unrestricted reduced-form covariance matrix. The solution is indexed by the number and composition of clusters and remains order- and scale-invariant.

The two benchmark identification schemes arise as corner solutions. When $C=N$, each variable forms its own cluster and $\boldsymbol{\Omega}_N=\boldsymbol{I}_N$, recovering the orthogonalized framework of \citet{francis2026principled}. When $C=1$, all variables belong to a single cluster and $\boldsymbol{\Omega}_1=\boldsymbol{\Sigma}$; the identity transformation $\boldsymbol{Q}_1^{-1}=\boldsymbol{I}_N$ then satisfies the constraint, recovering the generalized framework of \citet{koop1996impulse,pesaran1998generalized}.

\subsection{Lagrangian}

To characterize the optimal identification matrix, define the diagonal scaling matrix $\boldsymbol S_C=\operatorname{diag}\left(\frac{1}{\sqrt{\omega_{C,11}\sigma_{11}}},\ldots,\frac{1}{\sqrt{\omega_{C,NN}\sigma_{NN}}}\right).$ The objective function can then be written compactly as
\begin{align*}
\sum_{i=1}^N\frac{\mathbf e_i^\top\boldsymbol X_C\boldsymbol\Sigma\mathbf e_i}{\sqrt{\omega_{C,ii}}\sqrt{\sigma_{ii}}}=\operatorname{tr}\left(\boldsymbol S_C\boldsymbol X_C\boldsymbol\Sigma\right).
\end{align*}
The optimization problem therefore becomes
\begin{align*}
\max_{\boldsymbol X_C}\operatorname{tr}\left(\boldsymbol S_C\boldsymbol X_C\boldsymbol\Sigma\right)\qquad\text{s.t.}\qquad\boldsymbol X_C\boldsymbol\Sigma\boldsymbol X_C^\top=\boldsymbol\Omega_C.
\end{align*}
Let $\boldsymbol\Lambda$ denote a matrix of Lagrange multipliers. The corresponding Lagrangian is
\begin{align*}
\mathcal L=f(\boldsymbol X_C)+g(\boldsymbol X_C),
\end{align*}
where $f(\boldsymbol X_C)=\operatorname{tr}\left(\boldsymbol S_C\boldsymbol X_C\boldsymbol\Sigma\right)$ and $g(\boldsymbol X_C)=-\operatorname{tr}\left(\boldsymbol\Lambda^\top\left(\boldsymbol X_C\boldsymbol\Sigma\boldsymbol X_C^\top-\boldsymbol\Omega_C\right)\right).$

The Karush-Kuhn-Tucker (KKT) conditions \citep{karush1939minima,kuhn1951nonlinear} are:

\begin{align}
\frac{\partial\mathcal L}{\partial\boldsymbol X_C}
&=\frac{\partial f}{\partial\boldsymbol X_C}
+\frac{\partial g}{\partial\boldsymbol X_C}
=\mathbf0,\label{eq:lagrange1}\\
\frac{\partial\mathcal L}{\partial\boldsymbol \Lambda}
&=\boldsymbol X_C\boldsymbol\Sigma\boldsymbol X_C^\top
-\boldsymbol\Omega_C
=\mathbf0.
\label{eq:lagrange2}
\end{align}

To derive equation \eqref{eq:lagrange1}, we treat the objective term $f$ and the constraint term $g$ separately.

\subsubsection{Derivative of the objective function}

Recall that $f(\boldsymbol X_C)=\operatorname{tr}\left(\boldsymbol S_C\boldsymbol X_C\boldsymbol\Sigma\right).$ Taking the differential gives
\begin{align}
df=&\operatorname{tr}\left(\boldsymbol S_C\,d\boldsymbol X_C\,\boldsymbol\Sigma\right)\nonumber\\
=&\operatorname{tr}\left(\boldsymbol\Sigma\boldsymbol S_C\,d\boldsymbol X_C\right), \label{eq:of_1}
\end{align}
where the second equality uses the cyclic property of the trace to place the differential $d\boldsymbol X_C$ last.

By definition of the matrix derivative,
\begin{align}
df=\operatorname{tr}\left(\frac{\partial f}{\partial\boldsymbol X_C^\top} d\boldsymbol X_C\right). \label{eq:of_2}
\end{align}
Comparing equations \eqref{eq:of_1} and \eqref{eq:of_2} yields
\begin{align*}
\frac{\partial f}{\partial\boldsymbol X_C^\top}=&\boldsymbol\Sigma\boldsymbol S_C\\
\frac{\partial f}{\partial\boldsymbol X_C}=&\left(\boldsymbol\Sigma\boldsymbol S_C\right)^\top.
\end{align*}
Because both $\boldsymbol\Sigma$ and $\boldsymbol S_C$ are symmetric,
\begin{align}
\frac{\partial f}{\partial\boldsymbol X_C}=\boldsymbol S_C\boldsymbol\Sigma.\label{eq:of_final}
\end{align}

\subsubsection{Derivative of the constraint}
Next consider the constraint component of the Lagrangian,
\begin{align*}
g(\boldsymbol X_C)=-\operatorname{tr}\left(\boldsymbol\Lambda^\top\left(\boldsymbol X_C\boldsymbol\Sigma\boldsymbol X_C^\top-\boldsymbol\Omega_C\right)\right).
\end{align*}
Because $\boldsymbol\Omega_C$ is a prescribed target covariance matrix, it is fixed with respect to the optimization variable $\boldsymbol X_C$, and its differential is zero. Hence
\begin{align}
dg=&-\operatorname{tr}\left(\boldsymbol\Lambda^\top d\left(\boldsymbol X_C\boldsymbol\Sigma\boldsymbol X_C^\top\right)\right)\nonumber \\
=&-\operatorname{tr}\left(\boldsymbol\Lambda^\top d\boldsymbol X_C\,\boldsymbol\Sigma\,\boldsymbol X_C^\top+\boldsymbol\Lambda^\top \boldsymbol X_C\boldsymbol\Sigma\,d\boldsymbol X_C^\top\right)\nonumber\\ =&-\operatorname{tr}\left(\boldsymbol\Sigma\boldsymbol X_C^\top\boldsymbol\Lambda^\top d\boldsymbol X_C\right) -\operatorname{tr}\left(\boldsymbol\Sigma\boldsymbol X_C^\top\boldsymbol\Lambda d\boldsymbol X_C\right) \nonumber \\ 
=&-\operatorname{tr}\left(\boldsymbol\Sigma\boldsymbol X_C^\top(\boldsymbol\Lambda+\boldsymbol\Lambda^\top )d\boldsymbol X_C\right). \label{eq:of_3}
\end{align}

By definition of the matrix derivative,
\begin{align}
dg=&\operatorname{tr}\left(\frac{\partial g}{\partial\boldsymbol X_C^\top} d\boldsymbol X_C\right) \label{eq:of_4}
\end{align}
Comparing equations \eqref{eq:of_3} and \eqref{eq:of_4} yields
\begin{align}
\frac{\partial g}{\partial\boldsymbol X_C^\top}
=&
-\boldsymbol\Sigma\boldsymbol X_C^\top(\boldsymbol\Lambda+\boldsymbol\Lambda^\top )\nonumber \\
\frac{\partial g}{\partial\boldsymbol X_C}
=&
-(\boldsymbol\Lambda+\boldsymbol\Lambda^\top )\boldsymbol X_C\boldsymbol\Sigma. \label{eq:con_final}
\end{align}

Substituting equations \eqref{eq:of_final} and \eqref{eq:con_final} into equation \eqref{eq:lagrange1} gives
\begin{align}
\frac{\partial f}{\partial\boldsymbol X_C}+\frac{\partial g}{\partial\boldsymbol X_C}=&\mathbf0,\nonumber\\
\boldsymbol S_C\boldsymbol\Sigma-(\boldsymbol\Lambda+\boldsymbol\Lambda^\top )\boldsymbol X_C\boldsymbol\Sigma=&\mathbf0.\nonumber\\
\boldsymbol S_C\boldsymbol\Sigma=&(\boldsymbol\Lambda+\boldsymbol\Lambda^\top )\boldsymbol X_C\boldsymbol\Sigma\nonumber\\
\boldsymbol S_C=&(\boldsymbol\Lambda+\boldsymbol\Lambda^\top )\boldsymbol X_C. \nonumber
\end{align}

Thus the first-order condition links the unknown transformation matrix $\boldsymbol X_C$ linearly to the Lagrange multipliers. Together with the feasibility constraint $\boldsymbol X_C\boldsymbol\Sigma\boldsymbol X_C^\top=\boldsymbol\Omega_C$ from equation \eqref{eq:lagrange2}, it yields the KKT system
\begin{align}
\boldsymbol S_C
&=
(\boldsymbol\Lambda+\boldsymbol\Lambda^\top )
\boldsymbol X_C,
\label{eq:kkt1}\\
\boldsymbol\Omega_C
&=
\boldsymbol X_C
\boldsymbol\Sigma
\boldsymbol X_C^\top.
\label{eq:kkt2}
\end{align}

\subsubsection{Eliminating the Lagrange multipliers}

The first-order condition implies
\begin{align*}
\boldsymbol S_C=&(\boldsymbol\Lambda+\boldsymbol\Lambda^\top )
\boldsymbol X_C.\\
\boldsymbol S_C\boldsymbol X_C^{-1}=&\boldsymbol\Lambda+\boldsymbol\Lambda^\top .
\end{align*}
The right-hand side is symmetric by construction:
\begin{align*}
(\boldsymbol\Lambda+\boldsymbol\Lambda^\top )^\top
=
\boldsymbol\Lambda+\boldsymbol\Lambda^\top .
\end{align*}
Therefore,
\begin{align*}
\boldsymbol S_C\boldsymbol X_C^{-1}
=&
\left(\boldsymbol S_C\boldsymbol X_C^{-1}\right)^\top\\
=&
\boldsymbol X_C^{-\top}\boldsymbol S_C.
\end{align*}

We can therefore eliminate the Lagrange multipliers. The optimal transformation matrix $\boldsymbol X_C$ is characterized by the two matrix equations
\begin{align}
\boldsymbol S_C\boldsymbol X_C^{-1}
&=
\boldsymbol X_C^{-\top}\boldsymbol S_C,
\label{eq:KKT_final_11}\\
\boldsymbol\Omega_C
&=
\boldsymbol X_C\boldsymbol\Sigma\boldsymbol X_C^\top.
\label{eq:KKT_final_12}
\end{align}

\section{Identification in Correlation Space}
\label{sec:corrspace}

The KKT conditions derived above characterize the stationary points of the identification problem. To obtain an explicit solution, we must solve the matrix system in equations \eqref{eq:KKT_final_11} and \eqref{eq:KKT_final_12}. A direct solution is possible, but it involves cumbersome combinations of covariance matrices and their inverses. Reformulating the problem in correlation space separates marginal scales from dependence and yields a much cleaner quadratic matrix equation.

\subsection{Target Correlation Matrix}
We first separate marginal scales from dependence by rewriting the covariance matrices in terms of their standard deviations and correlations. Recall the matrix $\boldsymbol S_C$ introduced above to simplify the objective function. It can be written as
\begin{align}
 \boldsymbol S_C=\operatorname{diag}\left(\frac{1}{\sqrt{\omega_{C,11}\sigma_{11}}},\ldots,\frac{1}{\sqrt{\omega_{C,NN}\sigma_{NN}}}\right)=\boldsymbol D_C^{-1}\boldsymbol D_1^{-1}, \label{eq:S_C}
\end{align}
where
\begin{align*}
\boldsymbol D_1&=\operatorname{diag}\left(\sqrt{\sigma_{11}},\ldots,\sqrt{\sigma_{NN}}\right),\\\boldsymbol D_C&=
\operatorname{diag}\left(\sqrt{\omega_{C,11}},\ldots,\sqrt{\omega_{C,NN}}\right).
\end{align*}

Equation \eqref{eq:S_C} decomposes each covariance matrix into marginal standard deviations and a correlation matrix. This scaling is not imposed a priori; it emerges from the KKT conditions through $\boldsymbol S_C$, suggesting that the optimization is fundamentally over correlation structures rather than covariance structures. Define the corresponding decompositions,\footnote{The diagonal scaling of $\boldsymbol\Omega_C$ entails an additional normalization choice beyond the clustered correlation structure $\boldsymbol R_C$. We focus here on the latter and defer a systematic treatment of alternative shock normalizations and their implications to a subsequent version.}
\begin{align}
\boldsymbol\Sigma
&=
\boldsymbol D_1
\boldsymbol R_1
\boldsymbol D_1,\label{eq:D_1}\\
\boldsymbol\Omega_C
&=
\boldsymbol D_C
\boldsymbol R_C
\boldsymbol D_C.\label{eq:D_C}
\end{align}
Throughout the theoretical derivation, we assume that $\boldsymbol R_C$ is a symmetric positive definite block-diagonal correlation matrix. For empirical implementation, we construct $\boldsymbol R_C$ from the reduced-form correlation matrix $\boldsymbol R_1$ by retaining the empirical within-cluster correlations and setting all cross-cluster correlations to zero. Formally,
\begin{align*}
(\boldsymbol R_C)_{ij}
=
\begin{cases}
(\boldsymbol R_1)_{ij}, & \text{if } i \text{ and } j \text{ belong to the same cluster},\\
0, & \text{otherwise}.
\end{cases}
\end{align*}

Substituting equations \eqref{eq:D_1} and \eqref{eq:D_C} into the feasibility constraint \eqref{eq:KKT_final_12} gives
\begin{align*}
\boldsymbol\Omega_C
=&
\boldsymbol X_C\boldsymbol\Sigma\boldsymbol X_C^\top\\
\boldsymbol D_C\boldsymbol R_C\boldsymbol D_C
=&
\boldsymbol X_C\boldsymbol D_1\boldsymbol R_1\boldsymbol D_1\boldsymbol X_C^\top
\\
\boldsymbol R_C
=&
\boldsymbol D_C^{-1}\boldsymbol X_C\boldsymbol D_1\boldsymbol R_1\boldsymbol D_1\boldsymbol X_C^\top\boldsymbol D_C^{-1}\\
\boldsymbol R_C
=&
\boldsymbol Z
\boldsymbol R_1
\boldsymbol Z^\top  
\end{align*}
where 
\begin{align}
 \boldsymbol Z=\boldsymbol D_C^{-1}\boldsymbol X_C\boldsymbol D_1. \label{eq:Z}
\end{align}
The KKT conditions become
\begin{align}
\boldsymbol S_C\boldsymbol X_C^{-1}
&=
\boldsymbol X_C^{-\top}\boldsymbol S_C,
\label{eq:KKT_final_31}\\
\boldsymbol R_C
&=
\boldsymbol Z
\boldsymbol R_1
\boldsymbol Z^\top  \label{eq:KKT_final_32}
\end{align}

\subsection{Correlation Transformation}

Rewriting equation (\ref{eq:Z}) yields
\begin{align}
\boldsymbol X_C
=
\boldsymbol D_C\boldsymbol Z\boldsymbol D_1^{-1}, \label{eq:Z_bis}
\end{align}
so that
$\boldsymbol X_C^{-1}=\boldsymbol D_1\boldsymbol Z^{-1}\boldsymbol D_C^{-1}$ and 
$\boldsymbol X_C^{-\top}=\boldsymbol D_C^{-1}\boldsymbol Z^{-\top}\boldsymbol D_1$. Substituting these expressions and $\boldsymbol S_C=\boldsymbol D_C^{-1}\boldsymbol D_1^{-1}$ from equation \eqref{eq:S_C} into the first-order condition \eqref{eq:KKT_final_31} gives
\begin{align*}
\boldsymbol S_C\boldsymbol X_C^{-1}
&=
\boldsymbol X_C^{-\top}\boldsymbol S_C\\
\boldsymbol S_C\boldsymbol D_1\boldsymbol Z^{-1}\boldsymbol D_C^{-1}
&=
\boldsymbol D_C^{-1}\boldsymbol Z^{-\top}\boldsymbol D_1\boldsymbol S_C\\
\boldsymbol D_C^{-1}\boldsymbol D_1^{-1}\boldsymbol D_1\boldsymbol Z^{-1}\boldsymbol D_C^{-1}
&=
\boldsymbol D_C^{-1}\boldsymbol Z^{-\top}\boldsymbol D_1\boldsymbol D_C^{-1}\boldsymbol D_1^{-1}\\
\boldsymbol D_C^{-1}\boldsymbol Z^{-1}\boldsymbol D_C^{-1}
&=\boldsymbol D_C^{-1}\boldsymbol Z^{-\top}\boldsymbol D_C^{-1}\\
\boldsymbol Z^{-1}
&=
\boldsymbol Z^{-\top}\\
\boldsymbol Z
&=
\boldsymbol Z^\top  .
\end{align*}

The symmetry of $\boldsymbol Z$ provides the key simplification. The feasibility constraint \eqref{eq:KKT_final_32} reduces to the quadratic matrix equation
\begin{equation}
\boldsymbol Z\boldsymbol R_1\boldsymbol Z=\boldsymbol R_C.
\label{eq:Z_squared}
\end{equation}
Thus, in correlation space, SCO reduces to finding the symmetric positive definite transformation that maps $\boldsymbol R_1$ into $\boldsymbol R_C$. Because both matrices are symmetric positive definite, equation \eqref{eq:Z_squared} has a unique symmetric positive definite solution:
\begin{align*}
\boldsymbol Z
=\boldsymbol R_1^{-1/2}
\left(\boldsymbol R_1^{1/2}\boldsymbol R_C\boldsymbol R_1^{1/2}\right)^{1/2}
\boldsymbol R_1^{-1/2}.
\end{align*}
Substituting this solution into equation \eqref{eq:Z_bis} gives
\begin{align*}
\boldsymbol X_C
&=\boldsymbol D_C\boldsymbol Z\boldsymbol D_1^{-1}\\
\boldsymbol X_C
&=
\boldsymbol D_C
\boldsymbol R_1^{-1/2}
\left(\boldsymbol R_1^{1/2}\boldsymbol R_C\boldsymbol R_1^{1/2}\right)^{1/2}\boldsymbol R_1^{-1/2}\boldsymbol D_1^{-1}.
\end{align*}
Hence the optimal identification matrix is
\begin{empheq}[box=\fbox]{equation}
\label{eq:closed_form}
\boldsymbol Q_C^{-1}
=
\boldsymbol D_C\boldsymbol R_1^{-1/2}
\left(\boldsymbol R_1^{1/2}\boldsymbol R_C\boldsymbol R_1^{1/2}\right)^{1/2}
\boldsymbol R_1^{-1/2}\boldsymbol D_1^{-1}.
\end{empheq}

\subsection{Global Maximum}
Because the KKT conditions are necessary but not sufficient, equation \eqref{eq:closed_form} initially identifies only a stationary point. The next result establishes global optimality and uniqueness. The proof in Appendix~\ref{app:optimality} extends the argument of \citet{francis2026principled}: after whitening by the target covariance, every admissible identification matrix is parameterized by an orthogonal transformation of the solution, and no such transformation can increase the objective.

\begin{theorem}[Global optimality]
\label{thm:optimality}
Let $\boldsymbol R_1$ and $\boldsymbol R_C$ be symmetric positive definite. Then $\boldsymbol Q_C^{-1}$ in equation \eqref{eq:closed_form} is the unique global maximizer of the simultaneous clustered orthogonalization problem, and the maximized objective is
\begin{equation}
\sum_{i=1}^{N} \mathrm{corr}(u_i, \epsilon_i)
\;=\; \mathrm{tr}\!\left[ \left( \boldsymbol R_1^{1/2} \boldsymbol R_C\, \boldsymbol R_1^{1/2} \right)^{1/2} \right]
\;=\; \sum_{i=1}^{N} \mu_i^{1/2},
\label{eq:maxvalue}
\end{equation}
where $\mu_1, \ldots, \mu_N > 0$ are the eigenvalues of
$\boldsymbol R_1^{1/2} \boldsymbol R_C\, \boldsymbol R_1^{1/2}$.
\end{theorem}

Under the construction above, positive definiteness of $\boldsymbol R_C$ follows immediately: its diagonal blocks are principal submatrices of the positive definite matrix $\boldsymbol R_1$, and a block-diagonal matrix with positive definite blocks is itself positive definite.

\section{Corner Solutions}
\label{sec:corners}

The closed-form solution nests the two benchmark identification schemes as corner solutions. More generally, varying the number and composition of clusters generates a family of schemes in which structural shocks are orthogonal across clusters but may remain correlated within clusters. Generalized and fully orthogonalized identification therefore emerge as the two extremes of the broader clustered family.

\subsection{Full Orthogonalization ($C=N$)}

When every variable forms its own cluster, the structural covariance matrix reduces to the identity, $\boldsymbol\Omega_N=\boldsymbol I_N$, implying
\begin{align*}
\boldsymbol D_N
&=\boldsymbol I_N,
&&\text{Diagonal matrix of standard deviations},\\
\boldsymbol R_N
&=\boldsymbol I_N,
&&\text{Correlation matrix of structural shocks}.
\end{align*}
Substituting these expressions into the closed-form solution from equation (\ref{eq:closed_form}) gives
\begin{align*}
\boldsymbol Q_N^{-1}
&=\boldsymbol R_1^{-1/2}\left(\boldsymbol R_1^{1/2}\boldsymbol I_N\boldsymbol R_1^{1/2}\right)^{1/2}\boldsymbol R_1^{-1/2}\boldsymbol D_1^{-1}\\
&=\boldsymbol R_1^{-1/2}\boldsymbol R_1^{1/2}\boldsymbol R_1^{-1/2}\boldsymbol D_1^{-1}\\
&=\boldsymbol R_1^{-1/2}\boldsymbol D_1^{-1}.
\end{align*}
This is precisely the orthogonalization obtained by \citet{francis2026principled}:
\begin{align*}
\boldsymbol Q_N^{-1}
=
\boldsymbol R_1^{-1/2}
\boldsymbol D_1^{-1}.
\end{align*}
The maximized objective \eqref{eq:maxvalue} also specializes correctly when $C=N$: $\boldsymbol R_C=\boldsymbol I_N$, the $\mu_i$ are the eigenvalues of $\boldsymbol R_1$, and \eqref{eq:maxvalue} matches the maximal average correlation of \citet{francis2026principled}.

\subsection{No Orthogonalization ($C=1$)}

When all variables belong to a single cluster, $\boldsymbol\Omega_1=\boldsymbol\Sigma$, with
\begin{align*}
\boldsymbol D_1
&=\operatorname{diag}\!\left(\sqrt{\sigma_{11}},\ldots,\sqrt{\sigma_{NN}}\right),
&&\text{Diagonal matrix of standard deviations},\\
\boldsymbol R_1
&=\boldsymbol D_1^{-1}\boldsymbol\Sigma\boldsymbol D_1^{-1},
&&\text{Correlation matrix of reduced-form residuals}.
\end{align*}

Substituting these expressions into the closed-form solution from equation (\ref{eq:closed_form}) gives
\begin{align*}
\boldsymbol Q_1^{-1}
&=\boldsymbol D_1\boldsymbol R_1^{-1/2}\left(\boldsymbol R_1^{1/2}\boldsymbol R_1\boldsymbol R_1^{1/2}\right)^{1/2}\boldsymbol R_1^{-1/2}\boldsymbol D_1^{-1}\\
&=\boldsymbol D_1\boldsymbol R_1^{-1/2}\boldsymbol R_1\boldsymbol R_1^{-1/2}\boldsymbol D_1^{-1}\\
&=\boldsymbol D_1\boldsymbol I_N\boldsymbol D_1^{-1}\\
&=\boldsymbol I_N.
\end{align*}

Therefore,
\begin{align*}
\boldsymbol Q_1^{-1}=\boldsymbol I_N,
\end{align*}
which corresponds exactly to the generalized impulse response identification of \citet{koop1996impulse} and \citet{pesaran1998generalized}.  

The maximized objective \eqref{eq:maxvalue} likewise specializes correctly when $C=1$: $\boldsymbol R_C=\boldsymbol R_1$, the $\mu_i$ are the squared eigenvalues of $\boldsymbol R_1$, and \eqref{eq:maxvalue} equals $\mathrm{tr}(\boldsymbol R_1)=N$. Hence the average correlation attains its upper bound of one, as it must when $\boldsymbol Q_1^{-1}=\boldsymbol I_N$.

\section{Conclusion}

We develop simultaneous clustered orthogonalization (SCO), filling the missing case in a two-dimensional view of VAR identification: partial orthogonalization implemented simultaneously rather than sequentially. SCO retains the clustered structure of \citet{BDY1}, allowing shocks to remain correlated within clusters, while replacing sequential cross-cluster orthogonalization with simultaneous optimization. The resulting identification is invariant to cluster ordering.

The key mathematical result is especially simple in correlation space. The identification problem reduces to $\boldsymbol Z \boldsymbol R_1 \boldsymbol Z = \boldsymbol R_C$, which, for symmetric positive definite $\boldsymbol R_1$ and $\boldsymbol R_C$, has a closed-form symmetric positive definite solution. We show that the associated identification matrix is the unique global maximizer of the underlying correlation criterion. The framework nests generalized identification when all variables belong to one cluster and the full simultaneous orthogonalization of \citet{francis2026principled} when each variable forms its own cluster; intermediate partitions provide a family of partial orthogonalizations indexed by the number and composition of clusters.

For connectedness analysis, SCO preserves economically meaningful within-cluster contemporaneous dependence while removing the arbitrary cluster ordering inherited from sequential orthogonalization. Thus ordering dependence can be eliminated without imposing full orthogonality. Future versions will investigate alternative cluster structures and their empirical implications, including for stock-market connectedness.

\appendix

\section{Proof of Theorem \ref{thm:optimality}}
\label{app:optimality}
Throughout, $\boldsymbol X_C = \boldsymbol D_C \boldsymbol Z \boldsymbol D_1^{-1}$ denotes the solution in equation \eqref{eq:closed_form}. The matrix\\
$\boldsymbol Z = \boldsymbol R_1^{-1/2} \big( \boldsymbol R_1^{1/2} \boldsymbol R_C \boldsymbol R_1^{1/2} \big)^{1/2} \boldsymbol R_1^{-1/2}$
is symmetric positive definite and satisfies $\boldsymbol Z \boldsymbol R_1 \boldsymbol Z = \boldsymbol R_C$. The admissible set is $\mathcal{X}_C = \left\{ \boldsymbol X : \boldsymbol X \boldsymbol \Sigma \boldsymbol X^{\top} = \boldsymbol \Omega_C \right\}$, and the objective is $f(\boldsymbol X) = \mathrm{tr}(\boldsymbol S_C \boldsymbol X \boldsymbol \Sigma)$. The proof proceeds in two steps. First, after whitening by the target covariance, we show that the admissible set is parameterized by orthogonal transformations of $\boldsymbol X_C$. Second, we show that no such transformation increases the objective.

\begin{lemma}
\label{lem:param}
Let $\boldsymbol G = \boldsymbol D_C \boldsymbol R_C^{1/2}$, so that $\boldsymbol G \boldsymbol G^{\top} = \boldsymbol \Omega_C$. Then
$\boldsymbol X \in \mathcal{X}_C$ if and only if $\boldsymbol X = \boldsymbol G \boldsymbol U \boldsymbol G^{-1} \boldsymbol X_C$ for some
orthogonal $\boldsymbol U$, and $\boldsymbol U$ is unique given $\boldsymbol X$.
\end{lemma}

\begin{proof}
First, $\boldsymbol X_C \in \mathcal{X}_C$, because

$$
\boldsymbol X_C \boldsymbol \Sigma \boldsymbol X_C^{\top}
= \boldsymbol D_C \boldsymbol Z \boldsymbol D_1^{-1} \cdot \boldsymbol D_1 \boldsymbol R_1 \boldsymbol D_1 \cdot \boldsymbol D_1^{-1} \boldsymbol Z \boldsymbol D_C
= \boldsymbol D_C\, \boldsymbol Z \boldsymbol R_1 \boldsymbol Z\, \boldsymbol D_C = \boldsymbol D_C \boldsymbol R_C \boldsymbol D_C = \boldsymbol \Omega_C .
$$

Next, suppose $\boldsymbol X = \boldsymbol G \boldsymbol U \boldsymbol G^{-1} \boldsymbol X_C$ with $\boldsymbol U \boldsymbol U^{\top} = \boldsymbol I_N$. Because
$\boldsymbol G^{-1} \boldsymbol \Omega_C \boldsymbol G^{-\top} = \boldsymbol I_N$,

$$
\boldsymbol X \boldsymbol \Sigma \boldsymbol X^{\top}
= \boldsymbol G \boldsymbol U \boldsymbol G^{-1} \boldsymbol \Omega_C\, \boldsymbol G^{-\top} \boldsymbol U^{\top} \boldsymbol G^{\top}
= \boldsymbol G \boldsymbol U \boldsymbol U^{\top} \boldsymbol G^{\top} = \boldsymbol \Omega_C ,
$$

so $\boldsymbol X \in \mathcal{X}_C$. Conversely, let $\boldsymbol X \in \mathcal{X}_C$ and define
$\boldsymbol U = \boldsymbol G^{-1} \boldsymbol X \boldsymbol X_C^{-1} \boldsymbol G$. Then $\boldsymbol X = \boldsymbol G \boldsymbol U \boldsymbol G^{-1} \boldsymbol X_C$, and using
$\boldsymbol X_C^{-1} \boldsymbol \Omega_C \boldsymbol X_C^{-\top} = \boldsymbol \Sigma$,

$$
\boldsymbol U \boldsymbol U^{\top}
= \boldsymbol G^{-1} \boldsymbol X \big( \boldsymbol X_C^{-1} \boldsymbol \Omega_C \boldsymbol X_C^{-\top} \big) \boldsymbol X^{\top} \boldsymbol G^{-\top}
= \boldsymbol G^{-1} \boldsymbol X \boldsymbol \Sigma \boldsymbol X^{\top} \boldsymbol G^{-\top}
= \boldsymbol G^{-1} \boldsymbol \Omega_C \boldsymbol G^{-\top} = \boldsymbol I_N .
$$

\end{proof}

\begin{lemma}
\label{lem:K}
Let $\boldsymbol K = \boldsymbol G^{-1} \boldsymbol X_C \boldsymbol \Sigma \boldsymbol S_C \boldsymbol G$. Then
$\boldsymbol K = \boldsymbol R_C^{1/2} \boldsymbol Z^{-1} \boldsymbol R_C^{1/2}$, which is symmetric positive definite.
\end{lemma}

\begin{proof}
Diagonal matrices commute, so

$$
\boldsymbol X_C\, \boldsymbol \Sigma\, \boldsymbol S_C
= \boldsymbol D_C \boldsymbol Z \boldsymbol D_1^{-1} \cdot \boldsymbol D_1 \boldsymbol R_1 \boldsymbol D_1 \cdot \boldsymbol D_C^{-1} \boldsymbol D_1^{-1}
= \boldsymbol D_C\, \boldsymbol Z \boldsymbol R_1\, \boldsymbol D_C^{-1} .
$$

From $\boldsymbol Z \boldsymbol R_1 \boldsymbol Z = \boldsymbol R_C$ and the nonsingularity of $\boldsymbol Z$ we have
$\boldsymbol Z \boldsymbol R_1 = \boldsymbol R_C \boldsymbol Z^{-1}$, and therefore

$$
\boldsymbol K = \boldsymbol R_C^{-1/2} \boldsymbol D_C^{-1} \cdot \boldsymbol D_C\, \boldsymbol R_C \boldsymbol Z^{-1} \boldsymbol D_C^{-1} \cdot \boldsymbol D_C \boldsymbol R_C^{1/2}
= \boldsymbol R_C^{1/2}\, \boldsymbol Z^{-1} \boldsymbol R_C^{1/2} .
$$

The matrix $\boldsymbol Z^{-1}$ is symmetric positive definite, and congruence by the symmetric nonsingular matrix $\boldsymbol R_C^{1/2}$ preserves both properties.
\end{proof}

\noindent \emph{Proof of Theorem \ref{thm:optimality}.}
Let $\boldsymbol X \in \mathcal{X}_C$. By Lemma \ref{lem:param}, $\boldsymbol X = \boldsymbol G \boldsymbol U \boldsymbol G^{-1} \boldsymbol X_C$
for an orthogonal $\boldsymbol U$. By the cyclic property of the trace,

$$
f(\boldsymbol X)
= \mathrm{tr}\big( \boldsymbol S_C\, \boldsymbol G \boldsymbol U \boldsymbol G^{-1} \boldsymbol X_C\, \boldsymbol \Sigma \big)
= \mathrm{tr}\big( \boldsymbol U\, \boldsymbol G^{-1} \boldsymbol X_C\, \boldsymbol \Sigma\, \boldsymbol S_C\, \boldsymbol G \big)
= \mathrm{tr}( \boldsymbol U \boldsymbol K ) ,
$$

with $\boldsymbol K$ symmetric positive definite by Lemma \ref{lem:K}. Write
$\boldsymbol K = \boldsymbol Q \boldsymbol \Lambda_{\kappa} \boldsymbol Q^{\top}$, where $\boldsymbol Q$ is orthogonal and
$\kappa_i > 0$ for all $i$, and let $\boldsymbol M = \boldsymbol Q^{\top} \boldsymbol U \boldsymbol Q$, which is also
orthogonal. Then
  
$$ 
f(\boldsymbol X) = \mathrm{tr}\big( \boldsymbol M \boldsymbol \Lambda_{\kappa} \big)  
= \sum_{i=1}^{N} M_{ii}\, \kappa_i
\;\le\; \sum_{i=1}^{N} \kappa_i = \mathrm{tr}(\boldsymbol K),
$$
    
because no entry of an orthogonal matrix exceeds one in absolute value.
The bound is attained at $\boldsymbol U = \boldsymbol I_N$, so $\boldsymbol X_C$ is a global maximizer with
$f(\boldsymbol X_C) = \mathrm{tr}(\boldsymbol K)$.

For uniqueness, equality requires $M_{ii}=1$ for every $i$, because each $\kappa_i$ is strictly positive. Since the columns of an orthogonal matrix have unit norm, $M_{ii}=1$ forces all off-diagonal entries in column $i$ to vanish. Applying this argument to every column gives $\boldsymbol M=\boldsymbol I_N$, hence $\boldsymbol U=\boldsymbol I_N$ and $\boldsymbol X=\boldsymbol X_C$.

Finally, we compute the maximized objective. Diagonal matrices commute, so

$$
f(\boldsymbol X_C) = \mathrm{tr}\big( \boldsymbol S_C \boldsymbol X_C \boldsymbol \Sigma \big)
= \mathrm{tr}\big( \boldsymbol D_1^{-1}\, \boldsymbol Z \boldsymbol R_1\, \boldsymbol D_1 \big)
= \mathrm{tr}( \boldsymbol Z \boldsymbol R_1 ) .
$$

Substituting the closed form of $\boldsymbol Z$ and applying the cyclic property once more,

$$
\mathrm{tr}( \boldsymbol Z \boldsymbol R_1 )
= \mathrm{tr}\!\left[ \boldsymbol R_1^{-1/2} \big( \boldsymbol R_1^{1/2} \boldsymbol R_C \boldsymbol R_1^{1/2} \big)^{1/2}
\boldsymbol R_1^{1/2} \right]
= \mathrm{tr}\!\left[ \big( \boldsymbol R_1^{1/2} \boldsymbol R_C \boldsymbol R_1^{1/2} \big)^{1/2} \right]
= \sum_{i=1}^{N} \mu_i^{1/2} .
$$

The $\mu_i$ are strictly positive because $\boldsymbol R_1^{1/2} \boldsymbol R_C \boldsymbol R_1^{1/2}$ is a congruence of the positive definite matrix $\boldsymbol R_C$. $\qed$

\bibliographystyle{chicago}
\bibliography{biblio}

\end{document}